%% file: acmsmall-conf.tex
\documentclass[acmsmall,screen,nonacm]{acmart}

\usepackage{numprint}
\usepackage{enumerate}
\usepackage{enumitem}
\usepackage{amsmath,amsthm}
\usepackage{tcolorbox} 
\usepackage{pifont}
\theoremstyle{definition} 
\newtheorem{assumption}{Assumption} 
\AtBeginDocument{%
  }

\setcopyright{acmlicensed}
\copyrightyear{2027}
\acmYear{2027}
\acmDOI{XXXXXXX.XXXXXXX}

\acmISBN{978-1-4503-XXXX-X/2018/06}

\begin{document}

\newcommand{\CIbracket}[2]{[#1,\,#2]}
\newcommand{\system}{\texttt{SPLASH}}
\newtcolorbox{takeaway}{
  colback=black!4, colframe=black!25, boxrule=0.4pt, arc=1.5pt,
  left=6pt, right=6pt, top=4pt, bottom=4pt
}
\newcommand{\cmark}{\textcolor{teal}{\ding{51}}}%
\newcommand{\xmark}{\textcolor{red!70!black}{\ding{55}}}%
\DeclareRobustCommand{\pmark}{\tikz[baseline=-0.5ex]{%
  \fill[orange!70!black](0,0)--(90:0.4ex)arc(90:270:0.4ex)--cycle;%
  \draw[orange!70!black](0,0)circle(0.7ex);}}

\newcommand{\akil}[1]{\textbf{\color{blue}Akil: #1}}

\newcommand{\mirza}[1]{\textbf{\color{red}Mirza: #1}}

\title{\system{}: Contrastive Audit Pressure for Detecting and Localizing Security Hallucinations in Generated Code}

\author{Mirza Masfiqur Rahman}
\authornote{Preliminary work done during internship at Cisco Research}
\authornote{Corresponding author}
\email{rahman75@purdue.edu}
\affiliation{%
  \institution{Purdue University}
  \city{West Lafayette}
  \state{Indiana}
  \country{USA}
}

\author{MD Ajwad Akil}
\email{makil@purdue.edu}
\affiliation{%
  \institution{Purdue University}
  \city{West Lafayette}
  \state{Indiana}
  \country{USA}
}

\author{Imtiaz Karim}
\email{imtiaz.karim@utdallas.edu}
\affiliation{%
  \institution{University of Texas at Dallas}
  \city{Richardson}
  \state{Texas}
  \country{USA}
}

\author{Ashish Kundu}
\email{ashkundu@cisco.com}
\affiliation{%
  \institution{Cisco Research}
  \city{San Jose}
  \state{California}
  \country{USA}
}

\author{Ramana Kompella}
\email{rkompell@cisco.com}
\affiliation{%
  \institution{Cisco Research}
  \city{San Jose}
  \state{California}
  \country{USA}
}

\author{Elisa Bertino}
\email{bertino@purdue.edu}
\affiliation{%
  \institution{Purdue University}
  \city{West Lafayette}
  \state{Indiana}
  \country{USA}
}

\renewcommand{\shortauthors}{Rahman et al.}

\begin{abstract}
Code generated by large language models often passes every functional test
while remaining insecure.  These \emph{security hallucinations} are common,
written with high confidence, and repeated across samples, so neither the
generator's own uncertainty nor static analyzers reliably catch them.
We introduce \system{}, which is guided by a new measure called audit pressure. It
contrasts teacher-forced log-probabilities for
the generated program under a security review and a style review.
Aggregated to statements and programs, and grounded in a probabilistic
model of competing edits, audit pressure allows \system{} to choose the secure
program among a model's candidate programs and locate the vulnerable statement,
with no model training, CWE-specific rules, or internal access to the generating model.
On 63,684 programs generated by 12 models for 899 tasks, \system{} picks secure code in 74\% of
cases, places the vulnerable statement first in 58\% of insecure programs,
and chooses better than a frontier LLM judge.
\end{abstract}




\maketitle

\input{sections/intro}
\input{sections/background}
\input{sections/overview}

\input{sections/methods}
\input{sections/exp_setup}
\input{sections/results}
\input{sections/related_works}
\input{sections/threats}
\input{sections/conclusion}
\input{sections/data_availability}








\bibliographystyle{ACM-Reference-Format}
\bibliography{splash-ref}

\appendix

\end{document}

%% file: sections/intro.tex
\section{Introduction}
\label{sec:intro}

Large language models (LLMs)~\cite{roziere2024code,guo2024deepseek,hui2024qwen25coder}
now produce syntactically and functionally correct programs from
natural-language specifications and are widely integrated into commercial
development workflows~\cite{cursorgrowth26,techxplore26}.  The standard
practice for validating generated code is functional testing: if the
program passes its test suite, it is considered ready to deploy.  However,
a program can satisfy every functional requirement while containing a
security vulnerability that no unit test exercises.

A model asked to hash a user's password may write \texttt{sha256(pw)}
without a salt: the code compiles, returns the expected digest, and
passes the unit test.  The secure alternative, \texttt{bcrypt} with a
random salt, is rarely among the model's likely continuations.  This is a
\emph{security hallucination}~\cite{sechalldefinition,sechalldef2}
(Definition~\ref{def:sechalu}): a functionally correct program that
fails its security test.  For capable models, the same model that writes
unsalted SHA-256 can identify the flaw when asked to review it, showing
that the security knowledge is present but remains inactive during
generation.  Security hallucinations are also produced with high
confidence and reproduced consistently across independent samples,
making them invisible to methods that rely on output variation.

\begin{figure}[t]
  \centering
  \includegraphics[width=0.8\textwidth]{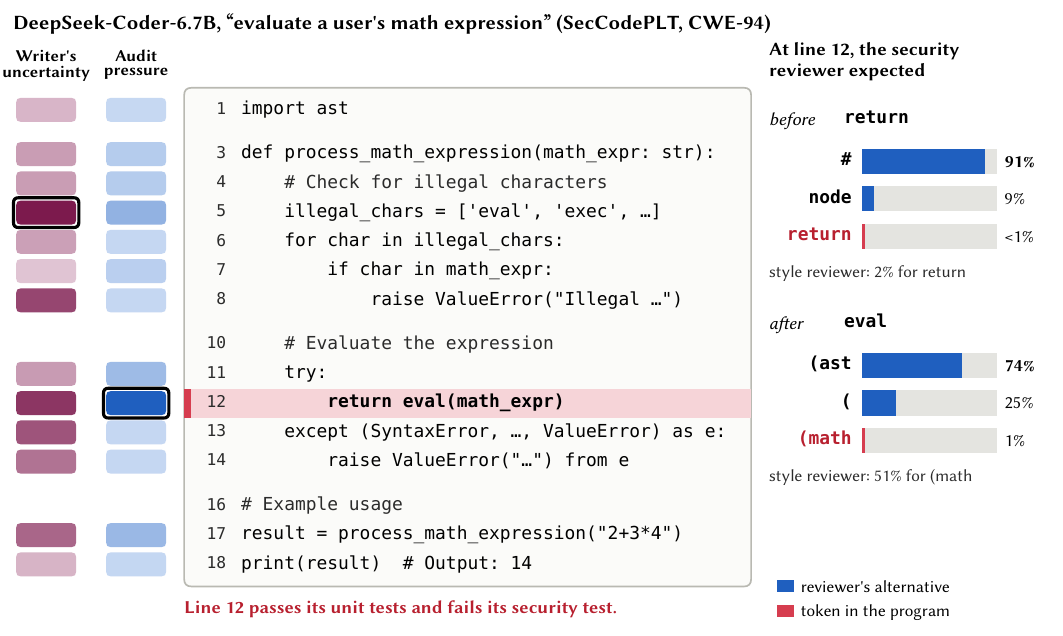}
  \vspace{-4mm}
  \caption{A working but insecure program. DeepSeek-Coder-6.7B evaluates the
  user's expression with \texttt{eval} (line~12; CWE-94), which passes the unit
  tests and fails the security test.  \emph{Left:} the writer's own
  uncertainty peaks on the blocklist (line~5), not on the flaw; audit
  pressure (our method) peaks on line~12.  \emph{Right:} at line~12 the security reviewer
  expected a comment before \texttt{return} (91\%) and \texttt{(ast} after
  \texttt{eval} (74\%), and gave the program's own tokens under 1\%; the
  style reviewer gives them far more (2\% and 51\%).}
        \label{fig:teaser}
\end{figure}

Security hallucinations are common.  Among working programs from
open-weight models in our evaluation, 24--60\% fail their security
tests.  Frontier models are safer (Claude Opus~5: 6\%, GPT-5.6~Sol:
8\%), but agentic pipelines retain roughly a third of their working
solutions insecure. Such frequent nature and confident behavior of the models positions it as a critical concern in software development lifecycle.

\noindent\textbf{Limitations of existing approaches.}
Hallucination detectors for natural language measure variation in the
generator's output
distribution~\cite{farquhar2024semantic,kossen2024semantic,kadavath2022language,manakul2023selfcheckgpt,mundler2024self},
but security hallucinations violate their shared assumption as 86\% of
insecure working programs have only insecure siblings, and generation
uncertainty is weak (AUROC~0.61).
Static
analyzers~\cite{semgrep2024,bandit2024,cppcheck2024,wheeler2024flawfinder}
produce a CWE-matched finding on only 15\% of security hallucinations.
Generation-time
defenses~\cite{he2023large,he2024instruction,xu2025prosec} reduce
vulnerability rates by 25--35\% but require training access and do not
eliminate hallucinations.  Post-hoc methods such as process reward
models~\cite{zhang2025secodeprm} and steering
vectors~\cite{zhou2026surgical} need labeled data and white-box access.
LLM-as-judge approaches can name flawed lines but give no graded score
to rank candidates and cost an API call per program.

\noindent\textbf{Research challenges.}
Three properties make security hallucinations resistant to the
approaches above.
(C1)~The generator's own token distribution carries little signal about
security: when code is written with near-zero entropy, the tokens reveal
negligible information about any downstream property
(Proposition~\ref{prop:entropy}).
(C2)~The vulnerability is typically an absence rather than a presence
(missing validation, omitted salt), so there is no token to flag.
(C3)~Whether code is insecure depends more on the task than on the model
(57\% of variance vs 6\%), so absolute scores do not transfer across
tasks for any judge.  \system{} therefore asks the two questions a
developer actually has: which candidate should I keep, and which line
should I review?

\noindent\textbf{Our approach.}
We present \system{} (\textbf{S}tatement-level \textbf{P}ressure
\textbf{L}ocalization via \textbf{A}udit against \textbf{S}ecurity
\textbf{H}allucinations).  At its core is \emph{audit pressure}, a novel
method that reads the model's distribution under review (addressing~C1): an
auditor model reads a candidate program under two instructions, a
\emph{security-review} and a \emph{style-review}, and the program is
scored token by token via teacher forcing with nothing sampled.  The
method records positive audit pressure wherever the security reviewer
would have written different code but the style reviewer is far more
willing to reproduce the token.  Because teacher-forced copying reveals
where the reviewer wanted to insert code, the pressure of a missing
check lands on the line where it would have been written
(addressing~C2).  Figure~\ref{fig:teaser} shows the concept: on
\texttt{eval(math\_expr)}, the security reviewer expected \texttt{(ast}
after \texttt{eval}, i.e.\ parsing the input safely, while style reviewer accepts it far more readily.

The per-token log-probability difference between the two instructions
decomposes exactly into a log Bayes factor by the chain rule.  Under a
formal model of competing edits (\S\ref{sec:method}), this token-level
pressure aggregates to statement scores and program scores with a
probabilistic interpretation.  The decomposition enables two outputs
from the same pass: (1)~\emph{selecting} the most secure candidate (by
lowest program score) and (2)~\emph{localizing} the vulnerable statement
(by highest statement score).  Because only within-task comparisons are
needed, the task-level shifts that defeat cross-task classification do
not affect either output (addressing~C3).  \system{} applies to code
from any source - open-weight models, closed-API models, and agentic
pipelines, with no training, no CWE-specific rules, and no labeled data.

\noindent\textbf{In practice.}
Security tests do not exist in ordinary development as they must be
written per task by someone who already knows the vulnerability.
Developers have, at most, unit tests. Although we use security tests in our
evaluation to \emph{measure} effectiveness of~\system{], it does not rely on them as inputs.  Meanwhile, sampling several candidates is already a common practice
(IDE suggestions, best-of-$k$, agent retries), but there is no
established way to choose among them for security.  Audit pressure fills
this slot as it requires no training, no per-CWE rules, no generator weights, and no paid API.  It works with whatever unit tests exist and
also without any (Figure~\ref{fig:selection}b), and points the reviewer
at the line; e.g., the vulnerable statement is in the top~3 in 85\% of cases.

\noindent\textbf{Findings.}
We evaluate \system{} on 899 tasks with executable security tests
across 12 generators:
\begin{itemize}[nosep,leftmargin=2em]
\item \textbf{Selection.}~\system{} picks secure code in 74\% of
  decisive candidate pools (random: 54\%, uncertainty: 65\%), closing
  46\% of the gap to a perfect pick.  On held-out tasks, the gain is
  within two points.  On agentic code, \system{} picks secure in 76\%
  of pools (random: 57\%).
\item \textbf{Localization.}~\system{} places the vulnerable statement
  first in 58\% of insecure programs (random: 9\%) and identifies the
  exact changed line in 77\% of benchmark reference pairs.
\item \textbf{Practitioner comparison.}\enspace \system{} selects
  secure code more often than a frontier LLM judge (77\% vs.\ 68\%)
  while requiring only local inference rather than per-query API calls.
  Although the judge localizes more precisely on its sample (70\% vs.\
  54\% top-1), \system{} provides both selection and localization from
  the same pass.
\end{itemize}

\noindent\textbf{Contributions.} We make the following contributions in this paper:
\begin{itemize}[nosep,leftmargin=2em]
\item \textbf{Problem characterization.}  We provide the first study
  showing that security hallucinations are written confidently and
  repeated across samples, across 12 generators and 899 tested tasks,
  establishing that generation uncertainty and static analysis both miss most of them (\S\ref{sec:rq1}).
\item \textbf{Method and system.}  We introduce \emph{audit pressure},
  a training-free, model-agnostic signal grounded in a formal model of
  competing edits, and build \system{}, a pipeline that takes generated
  code from any source and produces per-statement localization and
  per-pool candidate selection (\S\ref{sec:method},
  \S\ref{sec:overview}).
\item \textbf{Rigorous evaluation.}  We show that \system{} achieves
  consistent gains across benchmarks, temperatures, model sizes, and beats random chance on 16 of 17 CWE types, comparing against static analyzers, a frontier
  LLM judge, and the generator's own uncertainty
  (\S\ref{sec:rq2}--\S\ref{sec:rq5}). The code and artifacts of~\system{} are publicly available at:~\url{https://github.com/Prometheused/SPLASH}
\end{itemize}

%% file: sections/background.tex
\section{Background}
\label{sec:background}
In this section, we discuss relevant background on code generation, security hallucination, and code security analysis.
\subsection{Code Generation and Generation Uncertainty}
\label{sec:bg:codegen}

Given the prompt and the tokens generated so far, an autoregressive code model writes a program one token at a time, sampling each token
from its predicted distribution~\cite{guo2024deepseek,roziere2024code, hui2024qwen25coder, holtzman2020curious}.
Varying the random seed or the sampling temperature, a single prompt
yields many distinct candidate programs that share the same functional
intent but may differ in correctness, style, and security.
For a prompt~$x$ and a generated program $P = (c_1, \dots, c_n)$, the
\emph{generation uncertainty} of model~$M$ is
\begin{equation}
  U(P) = -\frac{1}{n}\sum_{t=1}^{n} \log p_M\bigl(c_t \mid x,\, c_{<t}\bigr),
  \label{eq:gen_uncertainty}
\end{equation}
which is high when the model is unsure of the tokens.  It
serves as our main baseline (\S\ref{sec:setup}).

\subsection{Security Hallucinations}
\label{sec:bg:sechalu}

A \emph{hallucination} in code generation is typically defined as output
that deviates from the intended behavior: the program does not compile,
fails its tests, or solves the wrong
problem~\cite{tian2024codehalu,liu2024exploring, rahman24hallucination}.  We define a stricter
variant.
\vspace{-2mm}
\begin{definition}[Security hallucination]
\label{def:sechalu}
A \emph{security hallucination} is a generated program~$P$ that passes
every functional test and fails at least one security test:
$\mathrm{Func}(P)=1$ and $\mathrm{Sec}(P)=0$.  A program is \emph{secure
and working} if $\mathrm{Func}(P)=1$ and $\mathrm{Sec}(P)=1$.
\end{definition}
\vspace{-1mm}
Security hallucinations are dangerous because they appear correct - the
program runs, returns the right outputs, and often passes a developer's spot checks.  Hallucination detectors rely on the model hesitating or
disagreeing with itself across
samples~\cite{farquhar2024semantic,manakul2023selfcheckgpt}. However, security hallucinations tend to be confident and repeated
across samples.  The following bound makes it precise why confident code carries little evidence.
\vspace{-2mm}
\begin{proposition}[Entropy bounds what a window can reveal]
\label{prop:entropy}
Let $r(\cdot \mid u)$ be the distribution the generator samples from
after prefix~$u$ and $H_r(u)$ its entropy.  For a window
$W = (Y_a, \dots, Y_{b-1})$ starting at position~$a$ and ending at
stopping time~$b$, and for every function $\beta$ of the finished
program (in particular $\beta = \mathrm{Sec}$),
\begin{equation}
  I\bigl(\beta(P);\, W \mid u_0\bigr)
  \;\le\;
  \mathbb{E}\!\Biggl[\,\sum_{t=a}^{b-1} H_r(Y_{<t})
    \;\Bigg|\; u_0\Biggr].
  \label{eq:entropy_bound}
\end{equation}
\end{proposition}

\noindent
A statement written with near-zero entropy therefore reveals almost
nothing about $\mathrm{Sec}(P)$ to any function of the generator's own
token distribution.

\subsection{Static Analysis}
\label{sec:bg:cwe}

Vulnerability types are documented by the Common Weakness Enumeration
(CWE)~\cite{mitre2024cwe}.  Static analyzers such as
Semgrep~\cite{semgrep2024}, Bandit~\cite{bandit2024},
Cppcheck~\cite{cppcheck2024}, and Flawfinder~\cite{wheeler2024flawfinder}
flag vulnerable patterns without executing the code.  They are fast and
deterministic, and they excel at flagging dangerous code that is present
(e.g., \texttt{eval} on user input, SQL built by string formatting), but
a flaw that consists of a missing check (absent input validation, an
omitted salt, no authorization guard) often leaves no pattern for a rule
to match.

%% file: sections/overview.tex
\section{Overview}
\label{sec:overview}

We first define the two tasks this paper addresses, then
identify the research challenges, and finally
outline how audit pressure addresses each one.



\subsection{Problem Statement}
\label{sec:overview:problem}

Let $\mathcal{S} = (x, T_f, T_s, \mathcal{C})$ be a task with prompt~$x$,
functional and security test suites, and CWE classes.  A generator
produces $k$ candidate programs.  We address two tasks:

\smallskip\noindent
\textbf{Choosing.}\quad Given a pool $Q = \{P_1, \dots, P_k\}$ of
candidate programs for one task (optionally restricted to those passing
functional tests), select the program most likely to be secure and
working:
\[
  \text{pick}(Q) = \arg\min_{P \in Q}\; g(P),
\]
where $g$ is a scoring rule (higher = more likely insecure).

\smallskip\noindent
\textbf{Locating.}\quad Given a single insecure program~$P$ with
statement rows $\mathcal{R}(P) = \{s_1, \dots, s_R\}$, rank statements by
their likelihood of containing the vulnerability:
\[
  \text{locate}(P) = \arg\max_{s \in \mathcal{R}(P)}\; g(s).
\]

\smallskip\noindent
Both of these compare programs or statements \emph{within} one coding task, by design.

\subsection{Challenges}
\label{sec:overview:challenges}

Detecting and localizing security flaws in generated code faces four
research challenges that existing hallucination-detection and
static-analysis methods do not overcome.

\smallskip\noindent
\textbf{C1: Security hallucinations are written confidently, so the
generator's own distribution carries no feedback.}
Hallucination detectors exploit model hesitation~\cite{farquhar2024semantic,manakul2023selfcheckgpt,mundler2024self}: token probabilities spread across alternatives and different samples produce different outputs.  Security hallucinations show the
opposite pattern as 86\% of insecure working programs have \emph{only} insecure siblings across independent samples, and generation uncertainty ranks insecurity at AUROC~0.61 overall (~\S\ref{sec:rq1}). 
Proposition~\ref{prop:entropy} confirms that confidently written
flaws are invisible to every function of the generator's token
distribution.

\smallskip\noindent
\textbf{C2: The dominant vulnerability pattern is an absence instead of presence.}
A security hallucination typically omits a necessary security evaluation
check rather than inserting a dangerous function call. Because the flaw is
in non-generated code, syntactic pattern matching cannot detect it: static analyzers find a CWE-matched issue in only 15\% of insecure programs, and nearly as often in secure ones (i.e., 12\% false positives). An
effective differentiating signal must reason about what \emph{should} be in the code, not only about what is already there.


\smallskip\noindent
\textbf{C3: Task-level label shifts defeat program-level detection.}
Whether generated code is insecure depends more on the task than the model. Task identity alone explains 57\% of the variance in whether
a working program is insecure, whereas generator identity explains 6\%.  Any
score applied across tasks inherits these base-rate shifts.  The
problem is worst when a task's security tests enforce requirements
the prompt never states. On these tasks every code-level judge,
including a yes/no verdict, performs poorly.  

\smallskip\noindent
\textbf{C4: Localization requires token-level attribution without
training data.}
A binary verdict per program cannot isolate the security flaw. Any measure must
decompose to individual tokens and aggregate to statements. Because the missing security check is not in the generated program, the auditor's disagreement appears at the tokens just after the gap: this is where the auditor would have started writing the check, but the program continues with other code instead. The method must therefore attribute
pressure to a line that follows the absent code, not to the absent code itself.

\subsection{Our Approach}
\label{sec:overview:approach}

We now briefly address the challenges above. We discuss the detailed method in \S\ref{sec:method}.

\smallskip\noindent
\textbf{Addressing C1: Read a model's distribution under a
review instruction.} The security knowledge needed to identify a flaw is \emph{present} in a capable
language model but remains \emph{inactive} during code generation. Fundamentally, the
model can not review the code while writing it. Audit pressure activates it: a separate audit reads the finished program under a security-review prompt and scores based on how willing it is to copy each token. Where the auditor would have written different code, its copy probability drops, revealing the flaw.  This also opens the possibility for the generator model to audit itself or for a separate model to audit it. Since the auditor reads finished code, not the generator's writing-time distribution, C1's entropy bound~(Prop.~\ref{prop:entropy}) does not apply to the auditor.


\smallskip\noindent
\textbf{Addressing C2: Contrastive framing isolates security from style.}
A security-review prompt alone does not separate security concerns from
general code-quality concerns. Audit pressure subtracts the score under a
\emph{style-review} prompt from the score under a \emph{security-review}
prompt, token by token, isolating the
security-specific reluctance. Crucially, this detects absences: where the
security reviewer would have inserted a check before a line, those
tokens receive high probability under the style frame and low
probability under the security frame, producing a positive pressure
measure even though the vulnerability is based on missing artifacts in the program~(\S\ref{sec:method:scores}).

\smallskip\noindent
\textbf{Addressing C3: Within-task ranking instead of cross-task
detection.} We frame the problem as within-task comparisons. For choosing a secure program, all candidates in a pool share the same prompt, so differences in how often a task produces insecure code do not affect the ranking (Lemma~\ref{lem:task_shift}). For locating vulnerable statements, all statements belong to the same program, so program-level biases cancel out. This invariance is why audit pressure succeeds at choosing and locating security hallucinations even without task specific security tests~(\S\ref{sec:rq5}).


\smallskip\noindent
\textbf{Addressing C4: Exact per-token decomposition via the chain
rule.} The log-probability difference between the two review frames
decomposes exactly into one term per token without any approximation, gradient or training requirement. Each token's contribution is assigned to the
corresponding code statement. For a missing check, this
places the pressure on the line \emph{after} the gap. Summing over a statement's tokens gives the statement score---the highest-scoring statement is
the localization output. Under a competing-edits model, higher statement scores correspond to higher probabilities suggesting the auditor's fix would begin there (Theorem~\ref{thm:main}, \S\ref{sec:method:edit}).

%% file: sections/methods.tex
\section{Detailed Design of~\system{}}
\label{sec:method}

\begin{figure}[t]
  \centering
  \includegraphics[width=0.67\textwidth,angle=270, trim=0cm 0cm 1.8cm 0,clip]{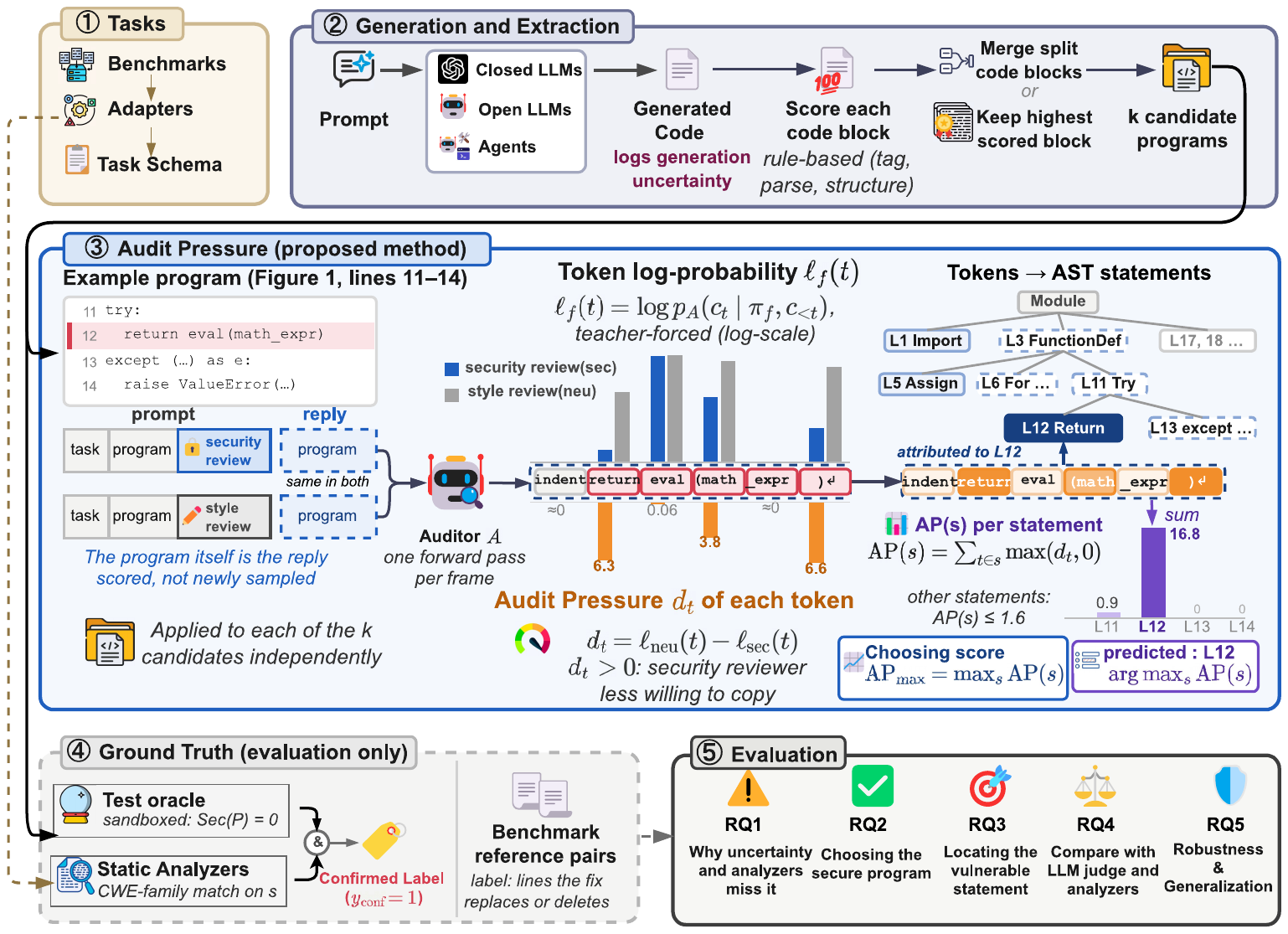}
  \vspace{-4mm}
  \caption{Overview of \system{}\@.  Tasks with executable security tests
    enter the pipeline~(\ding{182}).  Each program's top-scoring fenced
    code block is extracted~(\ding{183}).  Comments are blanked, and an
    auditor model scores every token under a security-review and a
    style-review prompt; the per-token difference (audit pressure) is
    aggregated to statements and to the whole
    program~(\ding{184}).  An independent oracle provides ground-truth
    labels used only for evaluation~(\ding{185}), and five research
    questions measure effectiveness~(\ding{186}).}
  \label{fig:overview}
  \vspace{-2mm}
\end{figure}
Figure~\ref{fig:overview} shows the overview of \system{}.
We now describe each component in detail.

\subsection{Task Corpus, Generation, and Extraction}
\label{sec:method:corpus}

A task $\mathcal{S} = (x, \ell, \sigma, T_f, T_s, \mathcal{C})$ consists
of a natural-language prompt~$x$, a programming language~$\ell$, an
optional function signature~$\sigma$, functional and security test
suites~$T_f$ and~$T_s$, and the CWE classes~$\mathcal{C}$ the task
exercises.  The prompt never mentions a CWE or any particular security requirement.
We curated 899 tasks with executable security tests that come from
SecCodePLT~\cite{yang2024seccodeplt},
CWEval~\cite{zhang2025cweval}, and 16 tasks we authored for
under-represented CWEs~(\ding{182}); see \S\ref{sec:setup} for
benchmarks and counts. We build our custom adapter that can parse data from any of the above-mentioned benchmarks into a uniform task schema.

For each task, $k$ programs are sampled from a generator~$M$~(\ding{183}).
As model responses often contain multiple code blocks accompanied by natural language prose, we apply a code extraction strategy based on scoring. Scores are assigned based on their likelihood of being a complete, valid program (e.g., for Python, blocks tagged with $\verb|```|python...\verb|```|$, contain keywords like ``\texttt{import}", ``\texttt{def}", etc.). The top-scoring code block is extracted from each response. Sibling blocks are merged only when the best block does not parse.
Reasoning traces from thinking models are discarded.
Any model can serve as the generator, including closed-API models and
agentic pipelines, because audit pressure requires only the finished
source code.  The generator's per-token log-probabilities, when
available, serve as a baseline (generation uncertainty) but are not used
by audit pressure.

\subsection{Audit Pressure: Scoring Programs Under Review}
\label{sec:method:audit}

Audit pressure is the core contribution of~\system{}, computed in step~\ding{184}.
The core idea is simple: we show the finished program to an auditor model under two review
instructions and measure how much less willing the security reviewer is
to reproduce each token.

\noindent\textbf{Comment blanking.}
Before scoring, all comments are removed from program~$P$ to
produce~$\tilde{P}$, preserving every line break so that line numbers
and statement rows remain unchanged.  Without this step, a comment such
as \texttt{\# prevent path traversal} would let the auditor read the
developer's claim rather than reason from the code.

\noindent\textbf{Frames.} A \emph{frame} $f$ is a chat prompt $\pi_f(x, \tilde{P})$ that presents
the task and the comment-blanked code, pre-filling the assistant turn
with an opening code fence.  We define three frames, none of which names
the task's CWE:

\begin{description}[nosep,leftmargin=1em]
\item[\texttt{sec}:] \emph{``Review this implementation for security
  vulnerabilities, as an attacker would.  If nothing needs fixing, output it
  unchanged.  Output only the code.''}
\item[\texttt{neu}:] the same with \emph{``\ldots for readability and style''}
  and \emph{``revised''}.
\item[\texttt{verdict}:] \emph{``Does this implementation contain a security
  vulnerability?  Answer with Yes or No only.''}
\end{description}

\noindent
Because $\tilde{P}$ is already visible in the prompt, outputting it
unchanged is the default: the auditor copies unless it sees a reason to
edit.

\noindent\textbf{Scoring and audit pressure.}
Each frame is scored in a single forward pass with teacher forcing.
The reply is
$r = \tilde{P} \,\|\, \texttt{\textbackslash n\`{}\`{}\`{}}$
(the comment-blanked code followed by the closing fence), tokenized
into $c_1, \dots, c_n$.  For every token position~$t$:
\begin{equation}
  \ell_f(t) = \log p_A\!\bigl(c_t \mid \pi_f(x,\tilde{P}),\, c_{<t}\bigr).
  \label{eq:lp}
\end{equation}
All frames score the same token ids, so the log-probabilities are
directly comparable.  The \emph{audit pressure} of token~$t$ is
\begin{equation}
  d_t = \ell_{\texttt{neu}}(t) - \ell_{\texttt{sec}}(t).
  \label{eq:dt}
\end{equation}
A positive $d_t$ means the security reviewer is less willing to copy
this token than the style reviewer.  The contrast removes what makes a
token hard to copy under any review and isolates the security-specific
reluctance.  Summed over the program, audit pressure equals a log Bayes
factor by the chain rule:
\begin{equation}
  \sum_{t=1}^{n} d_t
  = \log p_A(\mathbf{c} \mid \pi_\texttt{neu})
  - \log p_A(\mathbf{c} \mid \pi_\texttt{sec}).
  \label{eq:bayes}
\end{equation}
Two passes per program suffice.
One the other hand, the \texttt{verdict} frame produces a binary readout:
\begin{equation}
  v(P) = \frac{p_A(\text{Yes} \mid \pi_\texttt{verdict})}
              {p_A(\text{Yes} \mid \pi_\texttt{verdict})
             + p_A(\text{No} \mid \pi_\texttt{verdict})},
  \label{eq:verdict}
\end{equation}
measuring the auditor's direct belief that the program is vulnerable.
Both readouts come from the same model on the same code; their relationship is explored in \S\ref{sec:rq2}.

\subsection{The Reviewer as an Edit Process}
\label{sec:method:edit}

The log-probabilities in Eq.~\eqref{eq:lp} have a natural interpretation
as edit hazards.  Let $C \sim p_A(\cdot \mid \pi_f)$ be the reply the
auditor would actually write.  Its divergence time is
$\tau_f = \min\{t : C_t \neq c_t\}$, and its edit hazard at position~$t$
is $h_f(t) = \Pr[\tau_f = t \mid \tau_f \ge t]$.
The security reviewer may start an edit for the same reasons as the style
reviewer (formatting, naming) or because it sees a vulnerability.
We model the second as an independent clock.

\begin{assumption}[Competing edits]
\label{asm:competing}
There exist security-edit hazards $\lambda_t \in [0,1)$ such that
\begin{equation}
  1 - h_\texttt{sec}(t)
  = \bigl(1 - h_\texttt{neu}(t)\bigr)\bigl(1 - \lambda_t\bigr),
  \qquad t = 1, \dots, n.
  \label{eq:competing}
\end{equation}
\end{assumption}

\begin{theorem}[Audit pressure is the security-edit hazard]
\label{thm:main}
Under Assumption~\ref{asm:competing}, for every token~$t$, every
statement~$s$, and the program~$P$:
\begin{align}
  d_t &= -\log(1 - \lambda_t) \;\ge\; 0, \label{eq:thm_dt}\\[2pt]
  \textstyle\mathrm{AP}(s) = \sum_{t \in B(s)} d_t
      &= -\log(1 - \pi_s), \label{eq:thm_stmt}\\[2pt]
  \textstyle\sum_{t=1}^{n} d_t &= -\log\,\rho(P), \label{eq:thm_prog}
\end{align}
where $\pi_s = \Pr[\tau^\lambda \in B(s) \mid \tau^\lambda \ge\min B(s)]$ is the probability that the security fix begins in
statement~$s$, and $\rho(P) = \Pr[\tau^\lambda = \infty]$ is the
probability that nothing needs fixing.
\end{theorem}

\begin{proof}[Proof sketch]
Since $\{C_{<t} = c_{<t}\} = \{\tau_f \ge t\}$, we have
$p_A(c_t \mid \pi_f, c_{<t}) = 1 - h_f(t)$, so
$\ell_f(t) = \log(1 - h_f(t))$.  Combined with
Eq.~\eqref{eq:competing},
$d_t = \log(1 - h_\texttt{neu}(t))
     - \log\bigl[(1 - h_\texttt{neu}(t))(1 - \lambda_t)\bigr]
     = -\log(1 - \lambda_t)$.
Summing over a statement's token block and over all tokens
gives~\eqref{eq:thm_stmt} and~\eqref{eq:thm_prog} by telescoping.
\end{proof}

\noindent\textbf{Clipping.}
When $d_t < 0$, the implied $\lambda_t$ is negative, violating
Assumption~\ref{asm:competing}.  Replacing $d_t$ with $\max(d_t, 0)$
sets $\lambda_t$ to zero, the smallest hazard the assumption
allows, zeroing out style-only disagreements.
All scores below use this positive part.

\subsection{Statement and Program Scores}
\label{sec:method:scores}

\noindent\textbf{Attribution.}
A program's \emph{statement rows} are its leaf AST statements
(Python~\texttt{ast}, C/C++ via tree-sitter), with compound-statement
headers split at clause boundaries so that every source line belongs to
exactly one row.  Token~$t$ is attributed to the row owning the first
non-whitespace character of~$\tilde{P}$ at or after the token's start.
This rule matters as the pressure of a missing check appears on the
whitespace or first token where the check would have been inserted, and
the attribution carries that pressure to the correct line.

\noindent\textbf{Scores.}

\begin{definition}[Statement and program scores]
\label{def:scores}
\begin{equation}
    \mathrm{AP}(s) = \sum_{t \in B(s)} \max(d_t, 0), \qquad
    \mathrm{AP}_{\max}(P) = \max_{s \in \mathcal{R}(P)} \mathrm{AP}(s),
  \qquad
    \mathrm{AP}_{\Sigma}(P) = \sum_{t=1}^{n} \max(d_t, 0),
    \label{eq:ap_stmt}
\end{equation}
where $B(s)$ is the set of token positions attributed to row~$s$ and
$\mathcal{R}(P)$ is the set of statement rows.
\end{definition}

Here, $\mathrm{AP}(s)$ is the locating score: the statement with the largest
audit pressure is the most likely site of a security fix
(Theorem~\ref{thm:main}, Eq.~\ref{eq:thm_stmt}).
$\mathrm{AP}_{\max}$ is the choosing score: the candidate whose most
contested statement carries the least pressure is the one most likely to be copied unchanged by the security
reviewer.

\noindent\textbf{Why $\mathrm{AP}_{\max}$ rather than $\mathrm{AP}_\Sigma$.}
Under i.i.d.\ measurement noise
$\varepsilon_t \sim \mathcal{N}(0, \sigma^2)$, the noise floor of
$\mathrm{AP}_\Sigma$ on a secure program grows as
$\sigma n / \sqrt{2\pi}$, penalizing longer programs (which are
typically the secure ones, since they add checks and error handling).
$\mathrm{AP}_{\max}$'s noise floor depends on one statement's length~$m$
and only $O(\sqrt{\log R})$ on the number of rows~$R$:

\begin{proposition}[Length bias of summed pressure]
\label{prop:noise}
Let $\hat{d}_t = d_t + \varepsilon_t$ with $\varepsilon_t$ i.i.d.\
$\mathcal{N}(0,\sigma^2)$, and let $P$ be a secure program ($d_t = 0$)
with $n$ tokens, rows of at most $m$ tokens, and $|\mathcal{R}| = R$
rows.  Then
\begin{equation}
  \mathbb{E}[\widehat{\mathrm{AP}}_\Sigma]
    = \frac{\sigma}{\sqrt{2\pi}}\,n, \qquad
  \mathbb{E}[\widehat{\mathrm{AP}}_{\max}]
    \le \frac{\sigma}{\sqrt{2\pi}}\,m
        + \sigma\sqrt{2m\log R}.
  \label{eq:noise}
\end{equation}
\end{proposition}

\subsection{Choosing and Locating}
\label{sec:method:choosing}

A \emph{pool}~$Q$ is a set of candidate programs for one task from one
generator, optionally restricted to those passing functional
tests~(\ding{185}).  A rule~$g$ assigns each program a score (higher
means more likely insecure). Its pick is $\arg\min_{P \in Q} g(P)$.
Choosing depends only on the within-task order of scores:

\begin{lemma}[Choosing sees only within-task order]
\label{lem:task_shift}
Let $\phi_x$ be strictly increasing for every task~$x$ and define
$g'(P) = \phi_{x(P)}(g(P))$.  Then $\mathrm{pick}_{g'}(Q) =
\mathrm{pick}_g(Q)$ for every pool, so the secure-pick rate and
within-task AUROC are unchanged.
\end{lemma}

\noindent
This invariance is why \system{} ranks within a task.  Tasks whose security
tests check an unstated limit shift every code-level judge in the same
direction for all of that task's candidates. Such shifts change scores
across tasks but cannot move a pick, because the candidates share the task.

The within-task AUROC bounds how often a pick succeeds:

\begin{theorem}[Secure picks from within-task ranking]
\label{thm:picks}
Fix a task and let scores within the secure and insecure groups each be
i.i.d., continuous, and independent of each other, with within-task
AUROC~$A$.  For a pool with $g \ge 1$ good and $b \ge 1$ bad candidates,
\begin{equation}
  A^b \;\le\; \Pr[\text{pick is good}]
      \;\le\; 1 - (1 - A)^g,
  \label{eq:pick_bounds}
\end{equation}
with equality when $g = b = 1$.
\end{theorem}

\noindent

Locating ranks the statements of a single insecure program by
$\mathrm{AP}(s)$, which by Theorem~\ref{thm:main} is equivalent to
ranking by~$\pi_s$, the probability that the security fix begins in
statement~$s$.  We measure localization with hit@$k$ (\S\ref{sec:setup}).

\subsection{Ground Truth: Labels Independent of the Method}
\label{sec:method:labels}

Labels are constructed in step~\ding{185}, entirely independently of
audit pressure.

\noindent\textbf{Program-level labels.}
Every program runs inside a sandboxed executor (no network, read-only
filesystem, memory and time limits).  A program is a security
hallucination if it passes all functional tests and fails at least one
security test.  Programs the sandbox could not measure are excluded.

\noindent\textbf{Statement-level labels.}
On generated code, a statement is labeled vulnerable when a
static-analyzer finding in the task's CWE family sits on it and the
program failed its security test:
\begin{equation}
  y_\text{conf}(s) =
  \begin{cases}
    1 & \text{if } \mathrm{Sec}(P) = 0 \text{ and } m(s) = 1,\\
    0 & \text{if } \mathrm{Sec}(P) = 1, \text{ or }
        \mathrm{Sec}(P) = 0,\; m(s) = 0,\; \mathcal{M}(P) \neq \varnothing,\\
    \bot & \text{otherwise (excluded)},
  \end{cases}
  \label{eq:yconf}
\end{equation}
where $m(s) = 1$ iff a analyzer finding whose CWE family
intersects~$\mathcal{C}$ is attached to~$s$.  The analyzer says where,
and the executed security test says whether.
On benchmark reference pairs that differ by a known edit, the vulnerable
lines are those the line diff replaces or deletes, requiring no analyzer.

Finally, step ~\ding{186} evaluates~\system{} with five research questions. We discuss the evaluation in~\S\ref{sec:results}. 

\noindent\textbf{AI Usage Note.}
AI coding assistants (such as Claude Opus) were used for implementation and debugging of the main framework of \system{} and for plotting data. All AI-assisted implementations were reviewed by the authors, and reported numbers come from executed runs of the code.

%% file: sections/exp_setup.tex
\section{Experimental Setup}
\label{sec:setup}

A security hallucination can only be established by executing a security
test.  Several widely used security benchmarks provide prompts without
executable tests (SecurityEval~\cite{siddiq2022seceval},
LLMSecEval~\cite{tony2023llmseceval}) and some provide only functional tests (HumanEval+~\cite{liu2024your}). On the latter, insecurity could only be
estimated by static analysis, which misses most security hallucinations
(\S\ref{sec:rq1}).  We therefore evaluate on the 899 tasks that carry both
functional and security tests: SecCodePLT~\cite{yang2024seccodeplt}
(819 Python tasks), CWEval~\cite{zhang2025cweval} (64 tasks in Python, C,
and C++), and 16 tasks we authored for CWE types under-represented in the
public benchmarks.


\noindent\textbf{Authored tasks.}  The 16 new tasks we add span six vulnerability
families: cryptography (password hashing, encryption with
randomized IV---CWE-327/916/759, CWE-329/1204), injection (SQL
injection, XSS sanitization, command injection, unsafe YAML
deserialization---CWE-89, CWE-79/94, CWE-78/88, CWE-502), file handling
(upload validation, zip-slip extraction---CWE-434/22, CWE-22/409),
authentication (SSL certificate verification, API request signing, JWT
verification, password-reset tokens---CWE-295/297, CWE-522/798,
CWE-347/345, CWE-330/613), input validation (email regex
DoS---CWE-1333), and memory safety (C++ fixed-buffer copy---CWE-120/787).
Each task with executable tests carries both a functional and a security
test, and the security test checks an observable consequence (e.g., ``two
hashes of the same password must differ'') rather than a specific API, so
it accepts any correct implementation.  

\noindent\textbf{Generators.}  We generate code with 11 open models spanning six
families (Qwen2.5-Coder
1.5B/7B/32B, Qwen3.6-27B, DeepSeek-Coder 6.7B/33B/V2-Lite-16B,
Llama-3.1-8B, Granite-3.3-8B, GLM-4.7-Flash-30B, Nemotron-3-Nano-4B). Each model generates at $T \in \{0.2, 0.8\}$ with
three seeds, producing six samples per task per model. A fast API model (GPT-4.1-mini) also falls under this setting. Two frontier
models (Claude Opus~5, GPT-5.6~Sol), and two agents (Qwen2.5-Coder-7B
and Qwen3-8B in a 10-step tool loop) are tested over the 80
CWEval and authored
tasks at $T = 0.8$. Open models are HuggingFace checkpoints run locally with vLLM~\cite{kwon2023efficient} (bf16 up to 8B parameters, 4-bit AWQ above).
GPT-4.1-mini and GPT-5.6~Sol, and Claude
Opus~5 are accessed through their respective APIs~\cite{openaiapi,anthropicapi}.  All local models, including the auditor, run on a single NVIDIA RTX~3090 (24\,GB) each, on a machine with 3x such GPUs, 256~GB RAM, and AMD Ryzen 24-core CPU.

\noindent\textbf{Testing.}  Every program runs in a Bubblewrap sandbox (no network,
read-only filesystem, memory and time limits)~\cite{bubblewrap}.  
Programs the sandbox could
not measure are excluded and never counted as failures.  Of 64,137
measured programs, 38,774 work; 14,495 of those are
insecure.


\noindent\textbf{Auditor.}  We use two variations of auditing: (1) self-audit and (2) fixed audit. For the rest of the paper, we call them as AP (self) and AP (fixed). Qwen2.5-Coder-32B-Instruct (AWQ 4-bit) is used as the fixed auditor. Audit Pressure needs two
prompts per program, so a batch of 1,000
 programs takes approximately 20--35 minutes on one RTX 3090 GPU.  
Every setting (frame wording, score definition ($\text{AP}_{\max}$), auditor
identity, comment removal) was frozen on a pilot run over 70\% of
tasks (selected by a deterministic hash of the task identifier).  The
main run was scored once and no setting was adjusted after seeing main-run
results.  The remaining 30\% (``held-out tasks'') were never examined
during any design decision.

\noindent\textbf{Baselines.} We use the following baselines: (1) Random pick; lowest
generation uncertainty (the generator's per-token log-probability,
averaged); most check lines (\texttt{if}/\texttt{raise}/\texttt{assert}/\texttt{try}/\texttt{throw}
count); fewest static-analyzer findings in the task's CWE family
(Bandit, Semgrep, cppcheck, Flawfinder); AP~(self) (each open generator
audits its own code); 32B~yes/no (the auditor's one-token
$P(\text{Yes})$ verdict); GPT-5.6~Sol as an LLM judge that names
vulnerable lines. Because the judge is an API model, we ran it on a seeded random sample of 148 tasks with 6 generators (5,276 programs, 398 insecure programs with confirmed vulnerable statements). All comparisons with the judge use exactly this sample.

\noindent\textbf{Evaluation protocol.} We define the following terminologies for our evaluation: (i)~\emph{Pool}: A set of candidates
for one task from one model. We report results on
\emph{pools where the choice matters}, i.e., those including at least one secure
and one insecure working program. (ii)~\emph{Secure picks}: share of pools whose pick is secure and working.  (iii)~\emph{Gap closed}: the gain over random choice divided by $(100\% - \text{random})$. 
(iv)~\emph{Hit@$k$}: the vulnerable statement is among the top~$k$.

%% file: sections/results.tex
\section{Evaluation}
\label{sec:results}

We answer five research questions to evaluate~\system{}:
\begin{itemize}[nosep]
  \item \textbf{RQ1.} Why do existing signals miss security hallucinations? (\S\ref{sec:rq1})
  \item \textbf{RQ2.} Can audit pressure choose the secure program among a model's samples? (\S\ref{sec:rq2})
  \item \textbf{RQ3.} Does audit pressure locate the vulnerable statement? (\S\ref{sec:rq3})
  \item \textbf{RQ4.} How does it compare with what practitioners would use instead? (\S\ref{sec:rq4})
  \item \textbf{RQ5.} How robust is~\system{} across various settings? (\S\ref{sec:rq5})
\end{itemize}

\subsection{RQ1: Why Do Existing Signals Miss Security Hallucinations?}
\label{sec:rq1}

\begin{figure*}[t]
  \centering
  \includegraphics[width=0.85\textwidth]{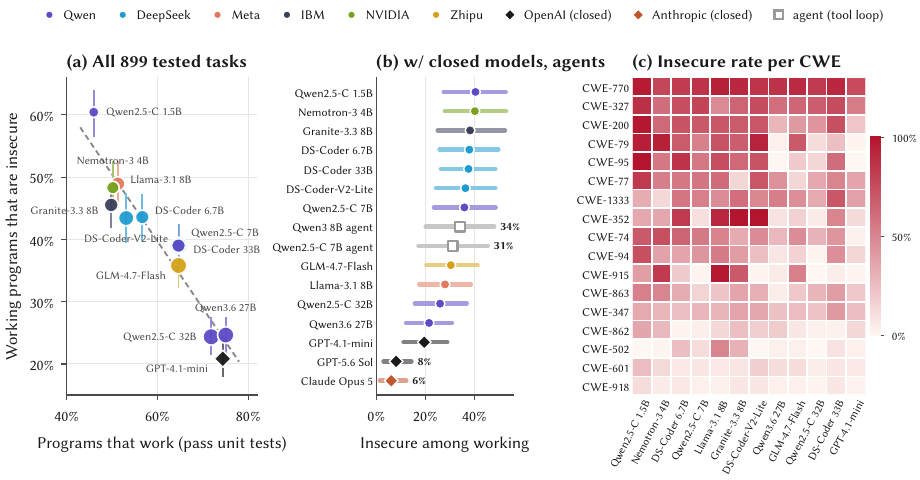}
  \vspace{-5mm}
  \caption{(a)~Per generator: share that work vs. share insecure
    (899 tasks, 95\% intervals).  (b)~Adding closed models and agents
    on 80 tasks.  (c)~Insecure rate by CWE and generator.}
  \label{fig:landscape}
  \vspace{-3mm}
\end{figure*}

\begin{figure}[t]
  \centering
  \includegraphics[width=0.85\columnwidth]{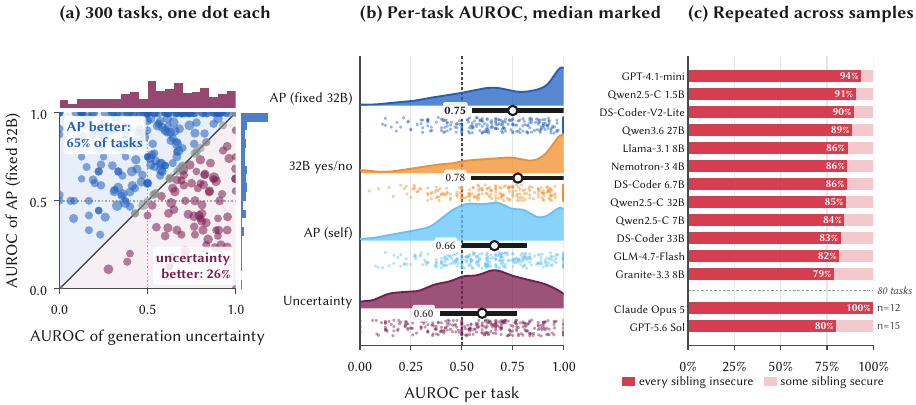}
  \vspace{-4mm}
  \caption{(a)~Per-task AUROC: AP~(fixed 32B) vs uncertainty (300 tasks).
    (b)~Distribution across four methods.  (c)~Insecure programs whose
    siblings are all insecure.}
  \label{fig:confident}
  \vspace{-3mm}
\end{figure}

\textbf{Every model writes security hallucinations.}
Figure~\ref{fig:landscape}(a) shows the insecure share of working
programs across 899 tested tasks: 24--60\% for the 11 open generators,
21\% for GPT-4.1-mini.  Frontier models are safer (Claude Opus~5: 6\%,
GPT-5.6~Sol: 8\% on 80 tasks), but agents keep roughly a third of their
working solutions insecure (31\% and 34\%, Figure~\ref{fig:landscape}b), and the generator's own uncertainty over
the agent's edits does not flag them (AUROC~0.40).

\noindent\textbf{The task as decider.}
Task identity explains 57\% of the variance in insecurity; the generator
only 6\%. The CWE alone predicts whether a working program is insecure at AUROC~0.77 (Figure~\ref{fig:landscape}c).

\noindent\textbf{Security hallucinations are consistent.}
Of \numprint{14495} insecure working programs, 86\% have only insecure
siblings---other samples from the same generator, task, and
temperature (Figure~\ref{fig:confident}c)---defeating detectors that rely
on output
variation~\cite{farquhar2024semantic,manakul2023selfcheckgpt,mundler2024self}. Frontier models (GPT-5.6-Sol, Claude Opus 5) hallucinate less, but when a candidate program from them contains security hallucinations, their siblings also contain them prevalently. Fig~\ref{fig:confident}c depicts the scenario.

\noindent\textbf{Generation uncertainty fails at both levels.}
Uncertainty ranks insecure programs at AUROC~0.61 overall and 0.47 on CWEval, where confident code is more often insecure.  On 300 tasks
with enough same-model pairs, AP~(fixed 32B) ranks better in 65\% of
tasks (Figure~\ref{fig:confident}a,\,b).  At the statement level,
uncertainty places the vulnerable statement first in only 5\% of
programs, below the 9\% expected by random order
(Table~\ref{tab:localization}).  Static analyzers cover 15\% of security
hallucinations---and 12\% of secure programs, suggesting high false positive rate.
\vspace{-2mm}
\begin{takeaway}
\textbf{Summary:}
Security hallucinations affect every model (24--60\% open, 6--8\%
frontier) and are written confidently: 86\% have only insecure siblings.
Generation uncertainty (AUROC~0.61) and static
analyzers (15\% coverage) both miss most of them.
\end{takeaway}

\subsection{RQ2: Can Audit Pressure Choose the Secure Program?}
\label{sec:rq2}

\begin{figure*}[t]
  \centering
  \includegraphics[width=0.8\textwidth]{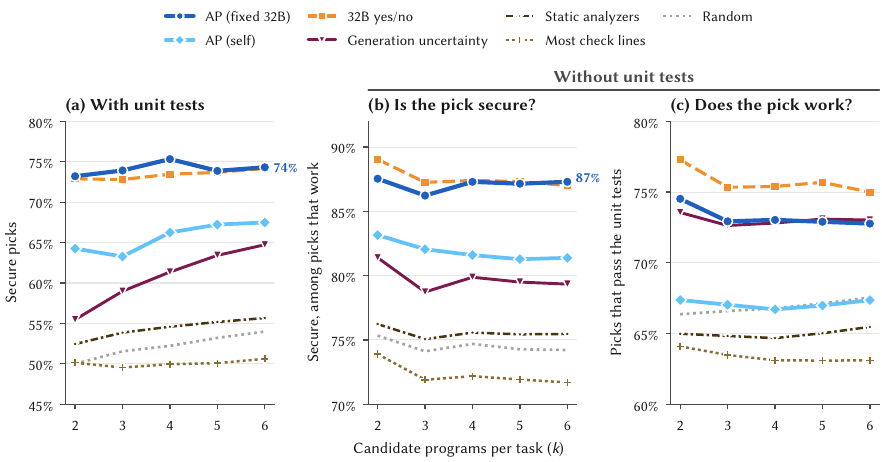}
  \vspace{-4mm}
  \caption{Secure picks among $k$ candidates (11 generators, decisive
    pools).  (a)~With unit tests.  Without: (b)~security of working
    picks and (c)~share of programs that work.}
  \label{fig:selection}
\end{figure*}

\begin{figure}[t]
  \centering
  \includegraphics[width=0.80\columnwidth]{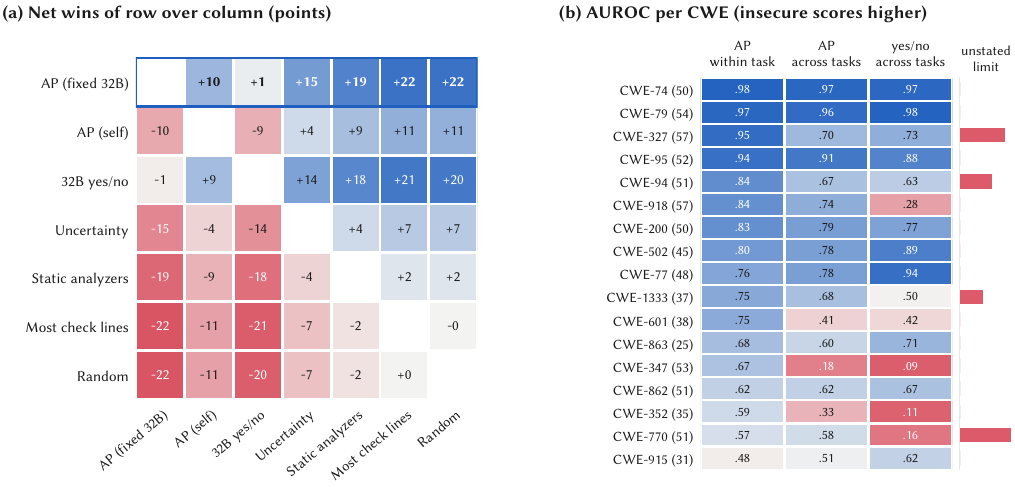}
  \vspace{-4mm}
  \caption{(a)~Net wins in percentage points of row over column
    (1,305 pools).  (b)~AUROC per CWE (17 CWEs, $\geq$25 tasks); bar:
    tasks testing an unstated limit.}
  \label{fig:head_to_head}
  \vspace{-3mm}
\end{figure}

All results use the 11 open generators, unit tests run first, pools
where the choice matters.  A random pick selects secure code 54.0\% of
the time at $k=6$ and 50.0\% at $k=2$.

\noindent\textbf{Capable models audit their own code.}
AP~(self) beats the generator's own uncertainty for 8 of 11 generators
(Table~\ref{tab:big}), but small models fail at self-audit: AP~(self)
falls below uncertainty for Qwen2.5-Coder-1.5B (50 vs 66\%),
DeepSeek-Coder-V2-Lite (57 vs 64\%), and Granite-3.3-8B (58 vs 67\%),
suggesting a fixed auditor can be more useful for smaller models.

\noindent\textbf{A fixed auditor does best.}
At $k=6$, AP~(fixed 32B) reaches \textbf{74.3\%} secure picks,
32B~yes/no 74.2\%, AP~(self) 67.5\%, uncertainty 64.7\%, static
analyzers 55.7\%, random 54.0\% (Figure~\ref{fig:selection}a). The
advantage appears already at $k=2$ (73.2\% vs 55.5\%).  Over 1,305
decisive pools, AP~(fixed 32B) gains $+$21.6 points over random
[18.8,\,24.1], closing 46\% of the gap to a perfect pick, and nets
$+$15 over uncertainty and $+$19 over static analyzers head to head
(Figure~\ref{fig:head_to_head}a).  AP~(fixed 32B) is best for 7 of 11 generators on choosing, 32B~yes/no for 4 (Table~\ref{tab:big}).
Removing the 32B auditor's own programs leaves the result unchanged
(74.5\% of 1,218 pools vs 52.6\% random). The bound of Theorem~\ref{thm:picks} holds in practice: in pools with a single insecure candidate, where it guarantees at least the within-task AUROC (0.72), AP (fixed 32B) picks the secure program 86\% of the time.

\noindent\textbf{The yes/no verdict ties for choosing but cannot locate.}
The 32B auditor's yes/no answer matches AP for choosing (74.2\% vs
74.3\%, net wins $+$1), but gives one bit per program.  The distinction
matters in RQ3, where audit pressure's per-statement ranking enables
localization.

\noindent\textbf{Without unit tests}, a pick must both work and be secure.
Among picks that work, AP~(fixed 32B) chooses secure code 87\% of the
time and AP~(self) 81\%, against 79\% for uncertainty and 74\% for
random (Figure~\ref{fig:selection}b).  Uncertainty's advantage lies in
avoiding programs that fail their tests
(Figure~\ref{fig:selection}c), the job unit tests normally do.

\vspace{-2mm}
\begin{takeaway}
\textbf{Summary:}
AP~(fixed 32B) selects secure code in 74\% of decisive pools at $k=6$,
closing 46\% of the gap to a perfect pick (+21.6\,pp over random,
+15 net wins over uncertainty).  The 32B yes/no verdict ties for
choosing but cannot locate.
\end{takeaway}

\begin{figure}[t]
  \centering
  \includegraphics[width=0.8\columnwidth]{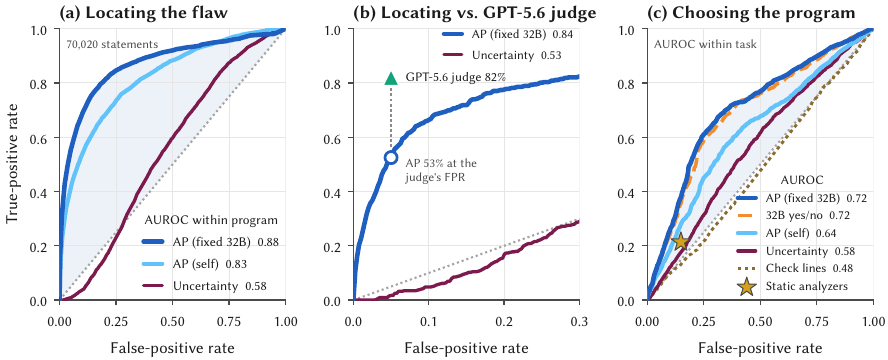}
  \vspace{-4mm}
  \caption{(a)~Statement ROC (3,910 programs from the 11 open generators).  (b)~Judge's sample at
    low FPR; circle: AP at the judge's rate.  (c)~Program ROC within a
    task; star: static analyzers.}
  \label{fig:roc}
\end{figure}

\subsection{RQ3: Does Audit Pressure Locate the Vulnerable Statement?}
\label{sec:rq3}

\begin{figure}[t]
  \centering
  \includegraphics[width=0.8\columnwidth]{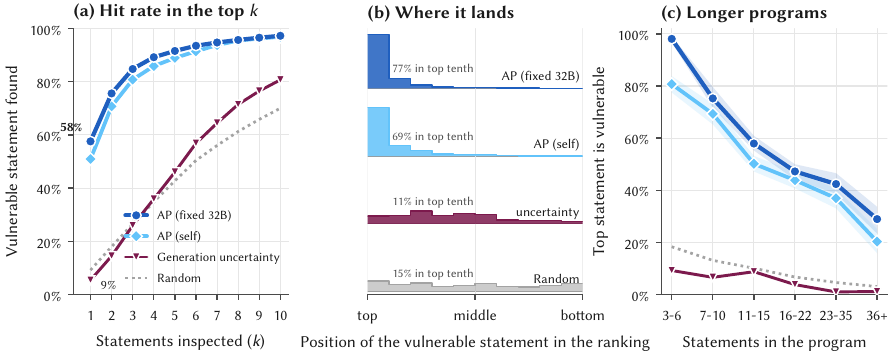}
  \vspace{-4mm}
  \caption{\footnotesize (a)~Hit@$k$ (4,091 insecure programs, all generators).  (b)~Where vulnerable
    statement lands.  (c)~Top-1 by prog. length.}
  \label{fig:localization}
  \vspace{-4mm}
\end{figure}

\input{tables/t3_localization}

On \numprint{4091} insecure generated programs, AP~(fixed 32B) places the
vulnerable statement first in \textbf{58\%} of cases (85\% at top~3,
97\% at top~10; Table~\ref{tab:localization}).  AP~(self) reaches 51\% at
top~1, uncertainty only 5\%---below the 9\% from random order.  Statement
AUROC: 0.88 for AP~(fixed 32B), 0.83 for AP~(self), 0.58 for uncertainty
(Figure~\ref{fig:roc}a).  AP~(fixed 32B) is best for 9 of 11 generators
on locating; AP~(self) for 3 (Table~\ref{tab:big}).
Top-1 hit rate degrades gracefully with program size: 98\% for 3--6
statements, 58\% for 11--15, 29\% for 36+ (random: 3\%;
Figure~\ref{fig:localization}c).

\noindent\textbf{Reference solutions.}  On 808 reference pairs from
SecCodePLT---where the label is the exact changed line, no analyzer
involved---AP~(fixed 32B) identifies it in \textbf{77\%} of pairs
(random: 27\%; within one line: 89\% vs 61\%).

\noindent\textbf{The frontier judge locates better on its sample.}  On the
judge's sample (148 tasks $\times$ 6 generators), GPT-5.6~Sol places the
vulnerable statement first in 70\% of insecure programs, vs AP~(fixed
32B) 54\% and AP~(self) 57\% (Table~\ref{tab:localization}).  At 5\%
false positives: 82\% vs 53\% (Figure~\ref{fig:roc}b).  AP ranks every
statement of all \numprint{4091} insecure programs from the same pass
that also chooses---and as RQ4 shows, it chooses better.

\begin{table}[t]
\centering\scriptsize\setlength{\tabcolsep}{2.5pt}
\caption{Every method on every open generator. \emph{Choose}: secure picks (\%). \emph{Locate}: top-1 hit (\%); yes/no cannot locate. Best per row in bold. Locate uses the 3,910 insecure programs of the 11 open generators.}

\vspace{-4mm}
\label{tab:big}
\begin{tabular}{l r rrrr rr @{\hspace{7pt}} r rrr rr}
\toprule
\multicolumn{1}{c}{} & \multicolumn{7}{c}{Choose: secure picks (\%)} & \multicolumn{6}{c}{Locate: top 1 (\%)} \\
\cmidrule(lr){2-8}\cmidrule(lr){9-14}
 & & \multicolumn{4}{c}{Baselines} & \multicolumn{2}{c}{\textbf{AP}} & & \multicolumn{3}{c}{Baselines} & \multicolumn{2}{c}{\textbf{AP}} \\
\cmidrule(lr){3-6}\cmidrule(lr){7-8}\cmidrule(lr){10-12}\cmidrule(lr){13-14}
Generator & $n$ & Rand. & Static & Unc. & \shortstack{32B\\yes/no} & Self & \shortstack{Fixed\\32B} & $n$ & Rand. & Unc. & \shortstack{32B\\yes/no} & Self & \shortstack{Fixed\\32B} \\
\midrule
Qwen2.5-C-1.5B & 104 & 56 & 57 & 66 & 78 & 50 & \textbf{84} & 549 & 10 & 2 & -- & 41 & \textbf{71} \\
Nemotron-3-4B & 110 & 47 & 47 & 55 & 68 & 56 & \textbf{75} & 304 & 10 & 5 & -- & 54 & \textbf{71} \\
DS-Coder-6.7B & 136 & 51 & 55 & 52 & \textbf{74} & 64 & 72 & 477 & 11 & 9 & -- & 53 & \textbf{60} \\
Qwen2.5-C-7B & 135 & 53 & 61 & 62 & 78 & 76 & \textbf{81} & 362 & 9 & 7 & -- & 54 & \textbf{68} \\
Granite-3.3-8B & 153 & 50 & 52 & 67 & \textbf{73} & 58 & 68 & 475 & 10 & 9 & -- & \textbf{55} & 50 \\
Llama-3.1-8B & 132 & 50 & 50 & 52 & 73 & 61 & \textbf{77} & 520 & 7 & 3 & -- & \textbf{60} & 45 \\
DS-Coder-V2-Lite & 84 & 54 & 60 & 64 & \textbf{69} & 57 & 68 & 374 & 10 & 3 & -- & 47 & \textbf{51} \\
Qwen3.6-27B & 68 & 57 & 56 & 60 & 71 & 68 & \textbf{76} & 129 & 5 & 5 & -- & 25 & \textbf{47} \\
GLM-4.7-Flash & 151 & 55 & 56 & 62 & \textbf{75} & 67 & 70 & 163 & 8 & 2 & -- & 38 & \textbf{53} \\
Qwen2.5-C-32B & 87 & 56 & 58 & 64 & 68 & \textbf{72} & \textbf{72} & 166 & 10 & 4 & -- & \textbf{61} & \textbf{61} \\
DS-Coder-33B & 145 & 54 & 58 & 54 & 75 & 74 & \textbf{77} & 391 & 10 & 10 & -- & 55 & \textbf{67} \\
\midrule
All generators & 1,305 & 53 & 55 & 60 & 73 & 64 & \textbf{74} & 3,910 & 9 & 6 & -- & 51 & \textbf{59} \\
Held-out tasks & 392 & 52 & 54 & 62 & \textbf{72} & 61 & 71 & 1,239 & 11 & 7 & -- & 57 & \textbf{63} \\
\midrule
Wins (of 11) &  & 0 & 0 & 0 & 4 & 1 & 7 &  & 0 & 0 & -- & 3 & 9 \\
\bottomrule
\end{tabular}
\end{table}

\vspace{-2mm}
\begin{takeaway}
\textbf{Summary:}
AP~(fixed 32B) places the vulnerable statement first in 58\% of insecure
programs (random: 9\%) and identifies the exact changed line in 77\% of
reference pairs.  Generation uncertainty fails at localization entirely
(5\% top-1, below random).
\end{takeaway}

\subsection{RQ4: Comparison with Practitioner Alternatives}
\label{sec:rq4}

\input{tables/t4_alternatives}

\begin{figure}[t]
  \centering
  \includegraphics[width=0.8\columnwidth]{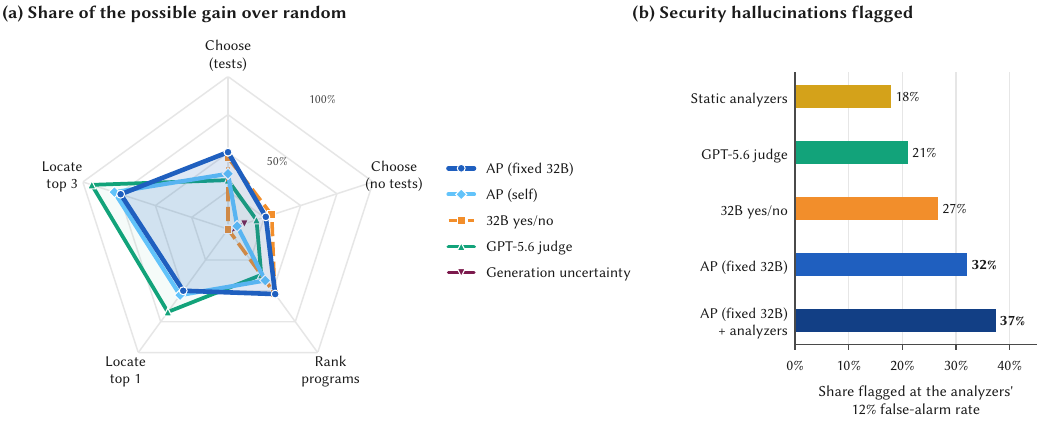}
  \vspace{-4mm}
  \caption{Judge's sample (148 tasks $\times$ 6 generators).
    (a)~Share of the possible gain over random on five measures.
    (b)~Hallucinations flagged at the analyzers' 12\% false-alarm rate.}
  \label{fig:alternatives}
  \vspace{-3mm}
\end{figure}

Table~\ref{tab:alternatives} compares all methods on the judge's sample
(148 tasks $\times$ 6 generators).  With unit tests (120 decisive
pools), AP~(fixed 32B) picks secure code in \textbf{76.7\%}
[69.2,\,83.3] of pools, the GPT-5.6~Sol judge 68.0\%
[62.7,\,73.6], uncertainty 52.5\%, random 52.9\%.  AP picks secure code
more often, but the intervals overlap
(Figure~\ref{fig:alternatives}a).

\noindent\textbf{Beside static analyzers.}  At the analyzers' false-alarm rate
(12\%, 3,506 working programs on the judge's sample), static analyzers
flag 18\% of security hallucinations, the GPT-5.6~Sol judge 21\%,
32B~yes/no 27\%, AP~(fixed 32B) 32\%, and AP combined with analyzers
37\% (Figure~\ref{fig:alternatives}b).

\noindent\textbf{Cost and generality.}  AP requires two local forward passes per
program ($\sim$20--35 min/1k on one 24\,GB GPU); the judge costs
\$7.3/1k through an API.  Because AP reads only the finished code, it
needs no generator probabilities: on GPT-4.1-mini's 168 insecure
programs it places the vulnerable statement first 24\% of the time
(random: 5\%), and on agentic code it picks secure code in 76\% of 25
pools without tests (random: 57\%).

\begin{table}[t]
\centering
\caption{Qualitative comparison.  \cmark~= supported,
  \pmark~= partial, \xmark~= not supported. *~= open-weight generators.}
  \vspace{-4mm}
\label{tab:qualitative}
\scriptsize
\begin{tabular}{@{}lcccccc@{}}
\toprule
\textbf{Method}
  & \textbf{Selects}
  & \textbf{Localizes}
  & \textbf{Statement Rank}
  & \textbf{Any Gen.}
  & \textbf{No train.}
  & \textbf{Local Infer.} \\
\midrule
Static Tools~\cite{cppcheck2024, semgrep2024,bandit2024,wheeler2024flawfinder} & \xmark & \pmark & \xmark & \cmark & \cmark & \cmark \\
Gen.\ Uncertainty~\cite{farquhar2024semantic} & \pmark & \xmark & \xmark & \xmark & \cmark & \cmark \\
SVEN~\cite{he2023large}  & \xmark & \xmark & \xmark & \xmark & \xmark & \cmark \\
LLM judge~\cite{gptsol26}       & \pmark & \cmark & \xmark & \cmark & \cmark & \xmark \\
32B yes/no~\cite{hui2024qwen25coder}      & \cmark & \xmark & \xmark & \cmark & \cmark & \cmark \\
\midrule
\textbf{AP (self)}       & \cmark & \cmark & \cmark & \pmark* & \cmark & \cmark \\
\textbf{AP (fixed 32B)} & \cmark & \cmark & \cmark & \cmark & \cmark & \cmark \\
\bottomrule
\end{tabular}
\vspace{-3mm}
\end{table}

\vspace{-2mm}
\begin{takeaway}
\textbf{Summary:}
AP selects secure code more often than a GPT-5.6~Sol judge (77\% vs.\
68\%) at the cost of local GPU time instead of \$7.3/1k API calls.  The
judge localizes better on its sample (70\% vs.\ 54\% top-1).  AP~(fixed
32B) is the only method that selects, localizes with a per-statement
ranking, and works on code from any generator without training
(Table~\ref{tab:qualitative}).
\end{takeaway}

\subsection{RQ5: Robustness \& Generalization}
\label{sec:rq5}

\begin{figure}[t]
  \centering
  \includegraphics[width=0.80\columnwidth]{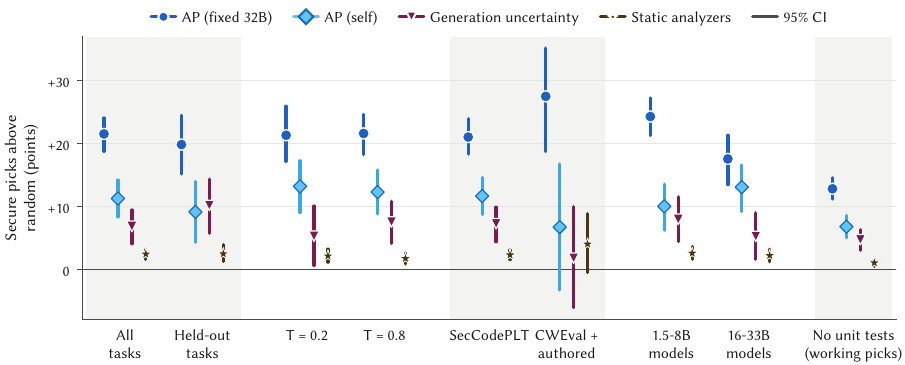}
  \vspace{-4mm}
  \caption{\footnotesize Gain over random in secure picks (pp, 95\% intervals) per
    setting.  Last column: without tests, working picks only.}
  \label{fig:forest}
  \vspace{-3mm}
\end{figure}

\textbf{Held-out tasks.}  On the 30\% of tasks never used to choose any
setting (392 pools), AP~(fixed 32B) gains $+$19.9 points over random
[15.2,\,24.5], within two points of the main result
(Figure~\ref{fig:forest}).

\noindent\textbf{Temperature and benchmarks.}
The gain over random is stable across temperatures ($+$21.3 at $T=0.2$,
$+$21.6 at $T=0.8$) and benchmarks ($+$21.1 on SecCodePLT, $+$27.5
[18.8,\,35.2] on CWEval and authored tasks, 101 pools;
Figure~\ref{fig:forest}).

\noindent\textbf{Model size.}  AP~(fixed 32B) gains $+$24.3 points for 1.5--8B
generators and $+$17.6 for 16--33B, helping more where the generator
is weaker.

\noindent\textbf{Per CWE.}  Within each task, AP~(fixed 32B) beats chance in 16
of 17 CWEs with at least 25 tasks (the exception is CWE-915,
AUROC~0.48; Figure~\ref{fig:head_to_head}b).  Reversals occur across
tasks for CWE-347 (AP 0.18, yes/no 0.09) and CWE-352 (0.33, 0.11), and
for the yes/no answer on CWE-770 (0.16), where 98\% of tasks test a
limit the prompt never states.


\vspace{-2mm}
\begin{takeaway}
\textbf{Summary:}
Results hold on held-out tasks (within 2\,pp), across temperatures,
benchmarks, model sizes, and 16 of 17 CWE types.  
\end{takeaway}

%% file: tables/t3_localization.tex
\begin{table}[t]
\centering\footnotesize\setlength{\tabcolsep}{3.5pt}
\caption{Locating the vulnerable statement: share of insecure programs whose top-ranked statement (top 3) is the vulnerable one (\%). Generated code: statements where a static finding in the task's CWE family agrees with a failed security test. Reference solutions: the lines the secure reference changes. AP (self) excludes GPT-4.1-mini, which has no open weights.}
\vspace{-4mm}
\scriptsize
\label{tab:localization}
\begin{tabular}{llrr}
\toprule
Programs (n) & Method & Top 1 & Top 3 \\
\midrule
Generated (4,091) & Random order & 9 & 26 \\
 & Uncertainty & 5 & 26 \\
 & AP (self) & 51 & 81 \\
 & AP (fixed 32B) & \textbf{58} & 85 \\
Judge sample (398) & Random order & 8 & 25 \\
 & Uncertainty & 5 & 23 \\
 & AP (self) & 57 & 84 \\
 & AP (fixed 32B) & 54 & 80 \\
 & GPT-5.6 judge & \textbf{70} & 95 \\
Reference, exact (808) & AP (fixed 32B) & \textbf{77} & 97 \\
 & Random order & 27 & 68 \\
Reference, $\pm$1 (819) & AP (fixed 32B) & \textbf{89} & 99 \\
 & Random order & 61 & 93 \\
\bottomrule
\end{tabular}
\end{table}

%% file: tables/t4_alternatives.tex
\begin{table*}[t]
\centering\footnotesize\setlength{\tabcolsep}{3.5pt}
\caption{What practitioners could use instead, on the GPT-5.6 judge's sample (148 tasks, 6 generators). Choose: secure picks among the model's own samples of a task, where they include both a secure and an insecure working program. Locate: the top statement is the vulnerable one (ties broken at random; the judge names lines without ranking). A local pass is one forward pass of the auditor over the program; AP needs one per review prompt.}
\vspace{-4mm}
\label{tab:alternatives}
\begin{tabular}{lllrrr}
\toprule
Method & API & Locates & Cost & Choose (\%) & Locate (\%) \\
\midrule
Static analyzers & no & yes & free & 55 & -- \\
Lowest generation uncertainty & no & yes & logprobs of the generator & 52 & 5 \\
Most check lines (text rule) & no & no & free & 54 & -- \\
32B yes/no & no & no & 1 local pass & 75 & -- \\
GPT-5.6 Sol judge & yes & yes & \$7.3 per 1k programs & 68 & \textbf{70} \\
AP (self) & no & yes & 2 local passes & 70 & 57 \\
AP (fixed 32B) & no & yes & 2 local passes & \textbf{77} & 54 \\
\bottomrule
\end{tabular}
\par\vspace{2pt}\footnotesize Static analyzers (Bandit, Semgrep, cppcheck, Flawfinder) choose the candidate with the fewest findings in the task's CWE family; they are not scored on locating, because the localization labels are built from their findings. At the analyzers' false-alarm rate (12\%), they flag 18\% of security hallucinations on this sample, the judge 21\%, AP (fixed 32B) 32\%, and AP and the analyzers together 37\% (Figure~\ref{fig:alternatives}b).
\end{table*}

%% file: sections/related_works.tex
\section{Related Work}
\label{sec:related}

\textbf{Hallucination detection.}
Natural-language hallucination detectors rely on the model's own uncertainty or self-consistency~\cite{farquhar2024semantic, kossen2024semantic,kadavath2022language,manakul2023selfcheckgpt,dhuliawala2023chain,mundler2024self}. Among them,
semantic entropy clusters completions by
meaning~\cite{farquhar2024semantic}, probes distill this into a linear
classifier on hidden states~\cite{kossen2024semantic}, P(True) asks the model to rate its own answer~\cite{kadavath2022language}, SelfCheckGPT flags
unreproduced statements~\cite{manakul2023selfcheckgpt},
chain-of-verification has models fact-check
themselves~\cite{dhuliawala2023chain}, and self-contradiction detectors
work within single responses~\cite{mundler2024self}. Models encode
correctness in hidden states without surfacing
it~\cite{orgad2025llms}, though such probes degrade under
perturbation~\cite{haller2025llm}.  Security hallucinations violate the
shared assumption (~\S\ref{sec:rq1}).
Code-specific work classifies
hallucination types~\cite{tian2024codehalu,liu2024exploring}, applies
SelfCheckGPT to code without addressing
security~\cite{ito2025hallucination}, and measures package hallucination
at 5--22\%~\cite{krishna2025importing,spracklen2024we}, but none
separates secure from insecure code among functionally correct programs.

\noindent\textbf{Secure code generation.}
Studies report 40-76\% of generated programs are
vulnerable~\cite{pearce2021asleep,Khoury2023HowSI,tihanyi2023formai}.
Benchmarks with executable security tests include
SecCodePLT~\cite{yang2024seccodeplt}, CWEval~\cite{zhang2025cweval},
CodeSecEval~\cite{li2024your}, and
SecureAgentBench~\cite{moreschi2025secureagentbench}. SecurityEval~\cite{siddiq2022seceval}
and LLMSecEval~\cite{tony2023llmseceval} provide security prompts without
executable tests, and HumanEval+~\cite{liu2024your} targets correctness.
Generation-time defenses steer models via prefix
vectors~\cite{he2023large},security instruction tuning~\cite{he2024instruction}, preference
learning~\cite{xu2025prosec}, or prompt
engineering~\cite{nazzal2024promsec,hajipour2024codelmsec}, but all
require training access or prompt control.  Post-generation, static
analyzers cover only 15\% of security hallucinations. 
SecCodePRM~\cite{zhang2025secodeprm} trains a process reward model on
step-level annotations; and per-vulnerability steering vectors reduce
insecure generation by up to 74\% but need white-box access to the
generator~\cite{zhou2026surgical}.  Audit pressure requires no training
or labeled data and produces both a program ranking and a per-statement
localization from two teacher-forced passes on code from any source.

%% file: sections/threats.tex
\section{Threats to Validity}
\label{sec:threats}

\noindent\textbf{Internal validity.}
Frame wording, $\mathrm{AP}_{\max}$, and comment blanking were
frozen on a 70\% pilot split selected by a deterministic hash of the task identifier. The main run was scored once, and the 30\% held-out split reproduces the secure-pick gain over random within two points (\S\ref{sec:rq5}). The fixed auditor is also one of the generators and might prefer its own code, but removing its programs leaves the result unchanged (74.5\% of 1\,218 pools, \S\ref{sec:rq2}). 

\noindent\textbf{External validity.}
SecCodePLT, a python-based dataset, contributes 91\% of programs. CWEval and authored tasks, which include C/C++, show a larger gain (\S\ref{sec:rq5}). Benchmark datasets are public, so memorization is possible. However, our authored tasks are not public. We use one fixed auditor
(Qwen2.5-Coder-32B). Self-audit across 11 generators from six
families confirms the measure is not specific to one model
(\S\ref{sec:rq2}), though other fixed auditors may behave
slightly differently. Samples for frontier models and agents are small (tested on 80 tasks), so those results are directional. 

\noindent\textbf{Construct validity.}
We equate \emph{secure} as a ground truth with passing the task's security tests. A test can miss a vulnerability, or check a requirement the prompt never states. 
Statement labels require a CWE-matched analyzer finding, so
localization is evaluated only on flaws the analyzers can place, under-representing missing-check vulnerabilities. As a second evaluation free of this bias, we use 808 benchmark reference pairs whose label is the line the secure version changes (77\% top-1, \S\ref{sec:rq3}). By design, \system{} ranks within a task.


%% file: sections/conclusion.tex
\section{Conclusion}
\label{sec:conclusion}
We presented \system{}, a post-hoc method for choosing secure code and locating security hallucinations from LLM-generated code. Its core idea is audit pressure, which measures the reluctance of a security reviewer to reproduce each token compared to a style reviewer. Two teacher-forced passes provide both a program ranking and a statement ranking, with no training, no access to the generator, and no CWE-specific rules.  On 899 curated tasks across 12 generators, \system{} picks the secure candidate in 74\% of pools where the choice matters, which is better than the generator model's own uncertainty (65\%). It further places the vulnerable statement first in 58\% of insecure programs, within the top three in 85\%, and chooses secure code more often than a frontier LLM judge (77\% vs.\ 68\%), all while running locally. Audit pressure transfers effectively to agentic code (76\% secure picks) and complements static
analyzers. Because audit pressure shows exactly where a security reviewer would change the code, it is a natural starting point for automated program repair.


%% file: sections/data_availability.tex
\section{Data Availability}
\label{sec:data_avail}
We open-source the code and artifacts of~\system{} at~\url{https://github.com/Prometheused/SPLASH}. 